\documentclass[prd,aps,amsfonts,eqsecnum,superscriptaddress,nofootinbib,longbibliography,notitlepage]{revtex4-1}

\usepackage{graphicx}
\usepackage{xcolor}
\usepackage{rotating}
\usepackage{amsmath,amssymb,graphics,amsthm,isomath}
\usepackage{bbm}
\usepackage{bm}
\usepackage{braket}

\usepackage[colorlinks=true, urlcolor=violet, linkcolor=blue, citecolor=red, hyperindex=true, linktocpage=true]{hyperref}

\allowdisplaybreaks

\newtheorem{thm}{Theorem}
\numberwithin{thm}{section}
\newtheorem{cor}[thm]{Corollary}
\newtheorem{conj}[thm]{Conjecture}
\newtheorem{lem}[thm]{Lemma}

\makeatletter
\renewcommand{\p@subsection}{}
\renewcommand{\p@subsubsection}{}
\makeatother

\usepackage{xcolor}
\usepackage{mathtools}

\newcommand{\ud}{\mathrm{d}}
\def\i{\text{i}}

\newcommand\mT{\mathcal{T}}

\newcommand\mR{\mathcal{R}}

\newcommand\mI{\mathcal{I}}
\newcommand\mU{\mathcal{U}}

\newcommand{\lr}[1]{\left( #1 \right)}
\newcommand{\mlr}[1]{\left[ #1 \right]}

\newcommand{\ii}{\mathrm{i}}

\newcommand{\PP}{\mathbb{P}}
\newcommand{\LL}{\mathcal{L}}
\newcommand{\norm}[1]{\left\lVert#1\right\rVert}

\newcommand{\abs}[1]{\left| #1 \right|}
\newcommand{\keto}[1]{\left| #1\right)}
\newcommand{\brao}[1]{\left( #1\right|}
\newcommand{\rarrow}{\quad \Rightarrow \quad}
\newcommand{\tU}{\widetilde{U}}
\newcommand{\tmU}{\widetilde{\mU}}

\newcommand{\comment}[1]{}

\newcommand{\tr}{\mathrm{tr}}

\newcommand\e{\mathrm{e}}

\usepackage{dsfont}

\begin{document}

\title{Frobenius light cone and the shift unitary}

\author{Chao Yin}
\email{chao.yin@colorado.edu}
\affiliation{Department of Physics and Center for Theory of Quantum Matter, University of Colorado, Boulder, CO 80309, USA}

\author{Andrew Lucas}
\email{andrew.j.lucas@colorado.edu}
\affiliation{Department of Physics and Center for Theory of Quantum Matter, University of Colorado, Boulder, CO 80309, USA}

\author{David T. Stephen}
\email{davidtstephen@gmail.com}
\affiliation{Department of Physics and Center for Theory of Quantum Matter, University of Colorado, Boulder, CO 80309, USA}
\affiliation{Department of Physics, California Institute of Technology, Pasadena, California 91125, USA}

\begin{abstract}
We bound the time necessary to implement the shift unitary on a one-dimensional ring, both using local Hamiltonians and those with power-law interactions.  This time is constrained by the Frobenius light cone; hence we prove that (for certain power law exponents) shift unitaries cannot be implemented in the same amount of time needed to prepare long-range Bell pairs.  We note an intriguing similarity between the proof of our results, and the hardness of preparing symmetry-protected topological states with symmetry-preserving Hamiltonians.
\end{abstract}

\date{\today}

\maketitle

\tableofcontents

\section{Introduction}
The past decade has seen an enormous advance in our understanding of the limitations of spatial locality in many-body physics \cite{Lieb1972,Hastings_koma,Nachtergaele_2006,Wang:2019jsj,Friedman:2022vqb}, including in systems with power-law interactions \cite{fossfeig,gong2017,Tran_2019_polyLC,else,strongly_gorshkov20,alpha_3_chenlucas,strictlylinear_KS,hierarchy20,Kuwahara_OTOC,Chen_1d21,tran2021} and bosons \cite{LRion,kuwahara2021liebrobinson,Yin:2021uio,Kuwahara:2022hlg,Faupin:2021nhk,Faupin_boson22,Lemm_boson23}: see \cite{ourreview} for a recent review.  These works often generalize the Lieb-Robinson Theorem \cite{Lieb1972}, which constrains commutators between space-time distant local operators: \begin{equation}
    \lVert [A_x(t),B_0]\rVert \le C(x,t). \label{eq:LRlikebound}
\end{equation}
By getting tight bounds on $C(x,t)$, we can show that e.g. the time it takes to perform state transfer between sites $x$ and 0 is at least as large as the smallest time $\tau$ for which $C(x,\tau)=2$ \cite{Bravyi2006}.  For many of the Lieb-Robinson bounds recently discovered, optimal protocols exist which demonstrate their tightness \cite{eldredge,Tran_GHZ21,yifan2021,Kuwahara:2022hlg}.

{ It is also known that Lieb-Robinson bounds can be lousy at bounding certain tasks: in other words, it might take much longer to achieve a desired goal than implied by (\ref{eq:LRlikebound}).} In theories with power-law interactions with a certain range of exponents, this was first emphasized in \cite{hierarchy20}, where it was noted that being able to implement a ``background-independent state transfer" that transfers \emph{Pauli operators}:
\begin{equation}
X_0(t) = X_r, \label{eq:frobeniusprotocol}
\end{equation}
cannot be done in the same amount of time as performing state transfer between sites 0 and $r$ { in one particular state: find unitary $U(t)$ such that \begin{equation}
    U(t)\left[|\psi\rangle_0 \otimes |\mathbf{0}\rangle_{\text{rest}}\right] = |\psi\rangle_r \otimes |\mathbf{0}\rangle_{\text{rest}}.\label{eq:xfer1state}
\end{equation}  Note that these two tasks are \emph{not} equivalent: fast state transfer protocols indeed rely on knowing that the state of other qubits is $|\mathbf{0}\rangle$, whereas any protocol obeying (\ref{eq:frobeniusprotocol}) does not care about the state of any particular qubit.  In this example above, in a one-dimensional chain it takes $t \gtrsim r/\log^2 r$ to achieve criterion (\ref{eq:frobeniusprotocol}), whereas (\ref{eq:xfer1state}) can be achieved in a time $t< r^\beta$ for any $\beta>0$. }  Since we have protocols that achieve time $t<r^\beta$ for task (\ref{eq:xfer1state}), and bounds that require $t>r/\log^2 r$ for task (\ref{eq:frobeniusprotocol}),  the authors of \cite{hierarchy20} said that there was a ``hierarchy of light cones", with differing ``light cones" (i.e. bounds on the time needed by an optimal protocol) existing for different tasks. 

\subsection{The shift unitary}

This paper is about another task which we will prove cannot be performed quickly perform in systems with power-law interactions {  than a naive application of (\ref{eq:LRlikebound}) would suggest}: implementing the shift unitary on a one-dimensional ring of $R$ sites.  We define the shift unitary $U_{\mathrm{sh}}$ as: \begin{equation}\label{eq:Ush_def}
    U_{\mathrm{sh}}^\dagger A_x U_{\mathrm{sh}} = A_{x+1}
\end{equation} for any single-qubit operator $A$.   Here and below, site indices $x$ are integers modulo $R$. Our goal is to constrain the time necessary to implement $U_{\mathrm{sh}}$ using a local Hamiltonian protocol.  Defining \begin{equation}\label{eq:U_def}
    U = \mathcal{T}\mathrm{e}^{-\mathrm{i}\int \mathrm{d}t H(t)},
\end{equation}
subject to locality and boundedness constraints on $H(t)$, we will see that it can take a parametrically long time in system size to obtain $U=U_{\mathrm{sh}}$. {   Clearly, (\ref{eq:LRlikebound}) does not give us any condition on the time it takes for $U^\dagger_{\mathrm{sh}}A_xU_{\mathrm{sh}}=A_{x+1}$, since this criterion can be achieved in finite time by using a Hamiltonian evolution that implements a SWAP gate on $x$ and $x+1$.  }

The shift unitary is an example of a quantum cellular automaton (QCA), which refers to any unitary operator with a strict light cone, meaning that $C(x,t)=0$ for large enough values of $x$ \cite{Farrelly2020}. In \cite{Gross2012}, the authors (GNVW) showed that all QCA can be written as a composition of a finite-depth quantum circuit and a translation operator. The net flow of information of generated by the translation is a discrete invariant captured by the GNVW index of a QCA. 
This index is formally defined in terms of operator algebras on an infinite lattice, but it can also be understood in finite systems using tensor network methods \cite{Cirac2017} or entropic measures \cite{Duschatko2018,Gong2021}. 
The index theory shows that any QCA with a non-zero index cannot be implemented as a finite-depth quantum circuit \cite{Gross2012}. Indeed, the simplest implementation of the shift unitary requires a circuit of SWAP gates whose depth grows linearly in system size. 

It is not \textit{a priori} clear whether the index theory of QCA carries over to the case where interactions are only approximately locality-preserving, corresponding to light cones with decaying tails. This case is important since it is true for generic time evolutions generated by quasi-local Hamiltonians. In \cite{QCA_LRB22}, the authors showed that the index theory is indeed well-defined even for approximately locality-preserving unitaries, such that QCA of non-zero index like $U_{\mathrm{sh}}$ cannot be implemented in finite time even with Hamiltonians having power-law decaying interactions (for large enough powers). However, the authors of \cite{QCA_LRB22} note that their technique is only applicable to strictly infinite-sized systems. One of the goals of this paper is to explain a few simple ways to prove the impossibility of easily implementing the shift on a finite ring. 


\subsection{Main result}
The set-up of this paper is as follows.  We consider a one-dimensional ring graph, where vertices are labeled by $\mathbb{Z}_N$.  We place a qubit on every site of the ring graph, such that the quantum mechanical Hilbert space $\mathcal{H} = (\mathbb{C}^2)^{\otimes N}$.   The distance between vertices $x,y\in\mathbb{Z}_N$ is \begin{equation}
    \mathsf{d}(x,y) = \min(|y-x|, N-|y-x|).
\end{equation}
For technical ease later, we focus on the case $N=4L$, but as our results are interesting in the large $N$ limit, this restriction could be easily relaxed.   The induced distance between two sets $S,S'$ is \begin{equation}
\mathsf{d}(S,S'):=\min_{x\in S, y\in S'}\mathsf{d}(x,y).
\end{equation}
The complement of a set $S$ is denoted by $S^{\rm c}$.  For operator $A \in \mathrm{End}(\mathcal{H})$, $\mathrm{supp}(A) \subseteq \mathbb{Z}_N$ denotes the sites on which it acts non-trivially.

The main result of this paper is the following theorem:

\begin{thm}\label{thm:powerlaw}
Let $H(t)$ be a 2-local Hamiltonian on the ring: 
\begin{equation}
    H(t) = \sum_{\lbrace x,y\rbrace \subset \mathbb{Z}_N} H_{xy}(t),
\end{equation} 
where $H_{xy}(t)$ acts non-trivially only on sites $x$ and $y$ (and identity elsewhere), and for some $0<K<\infty$: \begin{equation}
    \lVert H_{xy}(t)\rVert \le K \mathsf{d}(x,y)^{-\alpha}. \label{eq:powerlawdecay}
\end{equation}
We say that this Hamiltonian has power-law decaying interactions of exponent $\alpha$.\footnote{A strictly local Hamiltonian can therefore be thought of as power-law with arbitrarily large $\alpha$.} Then the two unitaries in (\ref{eq:Ush_def}) and (\ref{eq:U_def}) are far apart: \begin{equation}\label{eq:U-shift>}
    \norm{U-U_{\mathrm{sh}}} \ge 1/8,
\end{equation}so long as 
\begin{equation} 
        T \le \left\lbrace \begin{array}{ll} C^\prime L-C &\ \alpha\ge 4 \\ C^\prime L^{(\alpha-1)/3} - C &\ 3<\alpha<4 \\  C^\prime L^{(\alpha-2)(\alpha-1)/(2\alpha-3)-\epsilon}  &\ 2+1/\sqrt{2}<\alpha\le 3 \\ C^\prime L^{1/2-\epsilon} - C &\ 2\le \alpha<2+1/\sqrt{2} \\ C^\prime L^{(\alpha-1)/2} - C &\ 1<\alpha<2 \end{array}\right.. \label{eq:zoo}
    \end{equation}
 for arbitrarily small $\epsilon>0$, and for constants $0<C,C^\prime<\infty$ which do not depend on $L$, but may depend on $\epsilon$.
\end{thm}

The proof of Theorem \ref{thm:powerlaw} is provided in the proofs of Theorem \ref{thm:T>power} and Theorem \ref{thm:Frob_shift} in the body of the paper.  The proof of this main result has a bit of a distinct flavor from a standard Lieb-Robinson bound, where one worries about the time it takes to saturate (\ref{eq:LRlikebound}).  In a subtle contrast, to bound the time needed to implement $U_{\mathrm{sh}}$, it makes more sense to find the state that is \emph{hard} to prepare via shift.   Locality bounds then are used to constrain the time needed to prepare the shift on that particular state.  The zoo of different bounds in (\ref{eq:zoo}) is a consequence of the fact that we have to use somewhat different methods for differing $\alpha$: for $\alpha>2+1/\sqrt{2}$ we use Lieb-Robinson bounds, while for smaller $\alpha$ we use Frobenius bounds.  We conjecture, however, that the Frobenius light cone \cite{Chen_1d21} always bounds the hardness of implementing shift, suggesting that:
\begin{conj}
    (\ref{eq:zoo}) can be replaced by 
    \begin{equation}
        T \le \left\lbrace \begin{array}{ll} C^\prime L-C &\ \alpha\ge 2 \\ C^\prime L^{\alpha-1} - C &\ 1<\alpha<2 \end{array}\right..
    \end{equation}
\end{conj}

Lastly, we emphasize an intriguing analogy that we found somewhat non-trivial.   One way to understand that the shift operator is hard to implement is that it must move a generic operator on sites $1,2,\ldots, N/2$ to sites $2,3,\ldots,N/2+1$.  This average operator motion can occur only near sites 1 and $N/2$.  As a consequence, the shift operator is hard to implement because its implementation requires ``signaling" between the two boundaries of this operator.   An analogous idea underlies the proof that symmetry-protected topological (SPT) states are hard to prepare using symmetry-preserving Hamiltonian protocols \cite{SPT_lineardepth15}.  The ``Super2" formalism that we describe (which makes the above analogy more precise) can thus help to give a more physics-based intuition into the hardness of quantum computational tasks.

\section{Lieb-Robinson perspective}
We now turn to our proof of the main result of this paper, which will proceed in bits and pieces.  In this section, we first prove bounds on the hardness of implementing the shift operator by using Lieb-Robinson bounds.

\subsection{Local interactions}

For pedagogical purposes, we begin our discussion by relaxing the power-law decay assumption of (\ref{eq:powerlawdecay}).  Suppose that the Hamiltonian \begin{equation}
    H(t) = \sum_{x\in\mathbb{Z}_N} H_{x,x+1}(t) \label{eq:Hstrictlocal}
\end{equation}
with 
\begin{equation}\label{eq:Hxx<1}
    \max_x \norm{H_{x,x+1}(t)} \le 1.
\end{equation}
The discussion easily generalizes to any finite-range interaction.
For such systems, the following Lieb-Robinson bound holds: 
\begin{thm} \label{thm:LR} \cite{Lieb1972,ourreview} There exist constants $c_{\rm LR},\mu,v>0$, such that for any pairs of local operators $A,B$ with distance $r=\mathsf{d}(\mathrm{supp}A,\mathrm{supp}B)$, if $A(t)$ denotes Heisenberg evolution under $H(t)$ given in (\ref{eq:Hstrictlocal}),
    \begin{equation}\label{eq:LRB}
    \norm{[A(t), B]} \le c_{\rm LR} \norm{A}\norm{B} \e^{\mu(vt-r)}.
\end{equation}
\end{thm}

Our goal is then to prove: 
\begin{thm}\label{thm:shift}
Suppose \eqref{eq:LRB} holds. 
If the evolution time $T$ satisfies \begin{equation}\label{eq:t>N}
    vT\le L - C,
\end{equation}
for some constant $C$ (see \eqref{eq:Ce=} below), then (\ref{eq:U-shift>}) holds. 
\end{thm}

\begin{figure}
    \centering
    \includegraphics[width=.6\textwidth]{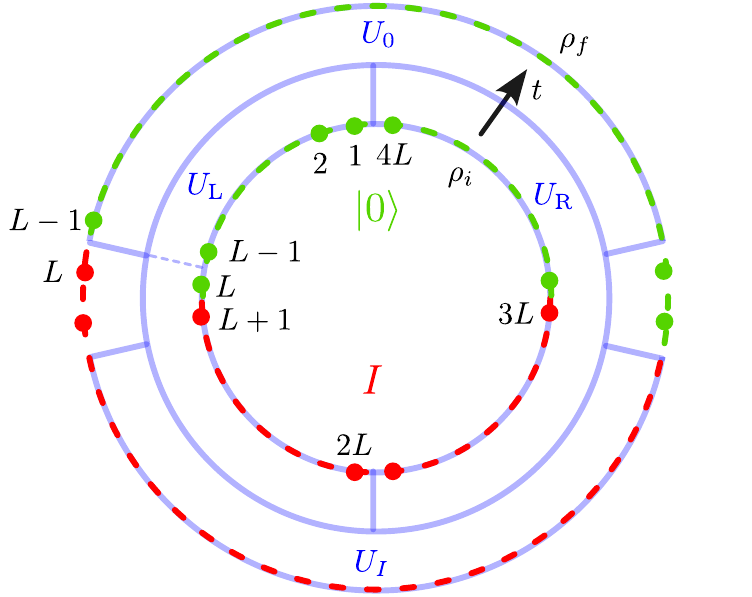}
    \caption{Sketch of the quantum circuit $\tU = \lr{ U_I \otimes U_0 }\cdot \lr{U_{\rm L}\otimes U_{\rm R} }$ that cannot achieve the shift operator. In the initial state \eqref{eq:rhoi}, the top half of the figure is $|0\rangle\langle 0|$ while the bottom half is $I$ on each site.}
    \label{fig:shift}
\end{figure}

\begin{proof}

To show \eqref{eq:U-shift>}, it suffices to find an approximation $\tU$ of $U$ such that \begin{subequations}\label{eq:U-tU-S}
    \begin{align}
        \norm{U-\tU} &\le \frac{1}{8}, \label{eq:U-tU} \\
        \norm{\tU-U_{\rm sh}} &\ge \frac{1}{4}, \label{eq:tU-S>}
    \end{align}
\end{subequations}
because from the triangle inequality, \begin{equation}\label{eq:U-S>tU}
    \norm{U-U_{\rm sh}}\ge \norm{\tU-U_{\rm sh}} - \norm{U-\tU} \ge \frac{1}{4}-\frac{1}{8}=\frac{1}{8}.
\end{equation}
We will choose $\tU$ as a quantum circuit (of very large gate size!) that approximations the Hamiltonian evolution, using ideas from \cite{Osborne_1d06,HHKL}. To this end, we first invoke the following Lemma to cut the coupling between site $0$ and $1$:

\begin{lem}\label{lem:HHKL}
If \eqref{eq:t>N} holds for constant \begin{equation}\label{eq:Ce=}
    C = 1+ \log \frac{16 c_{\rm LR}}{\mu v},
\end{equation}
there exists a unitary $U_0$ acting on region $\{2-L,3-L,\cdots,L-1\}$, such that \begin{equation}\label{eq:U=Uop}
    \norm{U-U_0 \cdot U_{\rm op} } \le \frac{1}{16},
\end{equation}
where \begin{equation}
    U_{\rm op} = \mathcal{T} \e^{-\ii \int^T_0 \ud t \mlr{H(t)-H_{01}(t)}},
\end{equation}
is generated by $H$ with open boundary. 
\end{lem}

We will prove this Lemma in Section \ref{sec:HHKL}. We then similarly cut $U_{\rm op}$ between sites $2L,2L+1$: \begin{equation}\label{eq:Uop-<}
    \norm{ U_{\rm op}- U_I\cdot (U_{\rm L}\otimes U_{\rm R}) }\le \frac{1}{16},
\end{equation}
where $U_{\rm L}$ ($U_{\rm R}$) is supported on the left (right) half chain $x=1,\cdots,2L$ ($x=2L+1,\cdots,4L$), and generated by the Hamiltonian terms inside the half chain. Analogous to $U_0$, $U_I$ is supported on sites $x=L+2,\cdots,3L-1$, where the meaning of subscripts $0,I$ will become clear later. 
To summarize, we choose \begin{equation}\label{eq:tU=}
    \tU = \lr{ U_I \otimes U_0 }\cdot \lr{U_{\rm L}\otimes U_{\rm R} },
\end{equation}
as shown by the four blue blocks in Fig.~\ref{fig:shift}, and the following corollary holds by the triangle inequality from Lemma \ref{lem:HHKL} and the analogous \eqref{eq:Uop-<}: 
\begin{cor}
    \eqref{eq:Ce=} leads to \eqref{eq:U-tU}.
\end{cor}

To conclude the proof, it remains to show \eqref{eq:tU-S>}. We
focus on a particular initial state $\rho_i$, and show it is hard to prepare the target final state \begin{equation}\label{eq:rhof=rhoi}
    \rho_f: = U_{\mathrm{sh}} \rho_i U_{\mathrm{sh}}^\dagger,
\end{equation}
using $\tU$.
The state $\rho_i$ is defined as follows. We choose identities on spin $L+1,\cdots,3L$, and $\ket{0}$ otherwise: \begin{equation}\label{eq:rhoi}
    \rho_i = \bigotimes_{x=1-L}^L \ket{0}_x\bra{0} \otimes \bigotimes_{x=L+1}^{3L} \frac{I_x}{2},
\end{equation}
so that \begin{equation}\label{eq:rhof}
    \rho_f= \bigotimes_{x=-L}^{L-1} \ket{0}_x\bra{0} \otimes \bigotimes_{x=L}^{3L-1} \frac{I_x}{2}.
\end{equation}
As a result, $U_I$ acts inside the identity region, for either $\rho_i$ or $\rho_f$; similar for $U_0$, which justifies the subscripts. For these states, we prove the following Lemma in Section \ref{sec:shift_rho}: 

\begin{lem}\label{lem:shift_rho} $\tU$ does not efficiently implement shift: 
\begin{equation}\label{eq:UrhoU-rho>}
    \norm{\tU \rho_i \tU^\dagger - \rho_f}_1 \ge 1/2.
\end{equation}
\end{lem}

With these two lemmas in  hand, \eqref{eq:UrhoU-rho>} then leads to \eqref{eq:tU-S>} because 
\begin{align}
    \norm{\tU \rho_i \tU^\dagger - U_{\mathrm{sh}} \rho_i U_{\mathrm{sh}}^\dagger}_1 &= \norm{\lr{\tU - U_{\mathrm{sh}}} \rho_i \tU^\dagger + U_{\mathrm{sh}} \rho_i \lr{\tU- U_{\mathrm{sh}}}^\dagger}_1 
    \le \norm{\lr{\tU - U_{\mathrm{sh}}} \rho_i \tU^\dagger}_1 + \norm{U_{\mathrm{sh}} \rho_i \lr{\tU- U_{\mathrm{sh}}}^\dagger}_1 \nonumber\\
    &\le \norm{\tU - U_{\mathrm{sh}}} + \norm{\tU- U_{\mathrm{sh}}^\dagger}=2\norm{\tU - U_{\mathrm{sh}}}.
\end{align}
We used triangle inequality in the first line, and H\"{o}lder's inequality $\norm{AB}_1\le \norm{A} \norm{B}_1$ with $\norm{\rho_i}_1=\norm{U}=1$ in the second line. This concludes the proof: we have shown \eqref{eq:U-tU}, and how \eqref{eq:U-shift>} follows from \eqref{eq:U-tU} and \eqref{eq:tU-S>}.
\end{proof}

\subsection{Circuit approximation of Hamiltonian evolution}\label{sec:HHKL}
\begin{proof}[Proof of Lemma \ref{lem:HHKL}]
For notational simplicity we assume $H$ is time-independent in this proof, because the time-dependent case generalizes immediately.
In the interaction picture, \begin{equation}
    U U_{\rm op}^\dagger = \e^{-\ii T H} \e^{\ii T (H-H_{01})} = \tilde{\mT} \e^{-\ii \int^T_0 H_{01}(t) \ud t}, \quad \mathrm{where} \quad H_{01}(t):= \e^{-\ii t (H-H_{01})} H_{01} \e^{\ii t (H-H_{01})},
\end{equation}
and $\tilde{\mT}$ is anti-time ordering. This can be verified by taking time derivative. 

$H_{01}(t)$ obeys Lieb-Robinson bound \eqref{eq:LRB}, so it is largely supported inside the region $S_r:=\{x:\mathsf{d}(x,\{0,1\})\le r\}=\{2-L,\cdots,L-1\}$ with $r=L-2$ (namely, $S_r$ is the initial support $\{0,1\}$ expanded by distance $r$). More precisely, Theorem \ref{thm:LR} implies (see e.g. Proposition 4.1 in \cite{ourreview}): \begin{equation}\label{eq:H-H'<LR}
    \norm{H_{01}(t) - \PP_r H_{01}(t)} \le c_{\rm LR} \e^{\mu(vt-L+1)},
\end{equation}
where we have used $\norm{H_{01}}\le1$ from \eqref{eq:Hxx<1}, and  \begin{equation}\label{eq:Pr=}
    \PP_r H_{01}(t):= \int_{S_r^{\rm c}} \ud V V^\dagger H_{01}(t) V
\end{equation}
with $V$ being a Haar random unitary acting on the complement $S_r^{\rm c}$. $\PP_r$ is a superoperator (linear transformation on the vector space of  operators) that selects the part $\PP_r H_{01}(t)$ whose support is contained inside $S_r$. Define 
\begin{equation}
    U_0 := \tilde{\mT} \e^{-\ii \int^T_0 \ud t\, \PP_r H_{01}(t) },
\end{equation}
supported in $S$. By the Duhamel identity, \begin{align}\label{eq:UUop-U0<}
    \norm{UU_{\rm op}^\dagger - U_0} &= \norm{-\ii \int^T_0 \ud t \tilde{\mT} \e^{-\ii \int^t_0 H_{01}(t_1) \ud t_1} \mlr{H_{01}(t) - \PP_r H_{01}(t)} \e^{-\ii \int^T_t \PP_r H_{01}(t_2) \ud t_2} } \nonumber\\
    &\le \int^T_0 \ud t \norm{H_{01}(t) - \PP_r H_{01}(t)} \le \frac{c_{\rm LR}}{\mu v} \e^{\mu(vT-L+1)},
\end{align}
where we have used \eqref{eq:H-H'<LR}. Because $\norm{U - U_0U_{\rm op}^\dagger}=\norm{UU_{\rm op}^\dagger - U_0}$, \eqref{eq:U=Uop} holds by setting \begin{equation}
    \frac{c_{\rm LR}}{\mu v} \e^{\mu(vT-L+1)} \le \frac{1}{16},
\end{equation}
which is equivalent to \eqref{eq:Ce=}. This lemma is, therefore, reminiscent of the HHKL algorithm \cite{HHKL}.
\end{proof}

\subsection{The quantum circuit approximation does not shift the particular state}\label{sec:shift_rho}
\begin{proof}[Proof of Lemma \ref{lem:shift_rho}]
    We manipulate the goal \eqref{eq:UrhoU-rho>} as follows: \begin{align}
    \frac{1}{2} &\le \norm{\tU \rho_i \tU^\dagger - \rho_f}_1 \nonumber\\
    &= \norm{\lr{U_{\rm L}\otimes U_{\rm R} } \rho_i \lr{U_{\rm L}\otimes U_{\rm R} }^\dagger - \lr{U_{\rm I}\otimes U_{\rm 0} }^\dagger \rho_f  \lr{U_{\rm I}\otimes U_{\rm 0}} }_1 \nonumber\\
    &= \norm{\lr{U_{\rm L}\otimes U_{\rm R} } \rho_i \lr{U_{\rm L}\otimes U_{\rm R} }^\dagger - U_{\rm 0}^\dagger \rho_f   U_{\rm 0} }_1. \label{eq:1/2<}
\end{align}
In the second line we have used $\norm{VOV^\dagger}_1 = \norm{O}_1$ for any unitary $V$. In the last line we used the fact that $U_I$ acts trivially on the identity region in $\rho_f$. Taking partial trace in the last line of \eqref{eq:1/2<} and using monotonicity of trace distance: $\norm{\tr_{\rm R}(\rho-\rho')}_1\le \norm{\rho-\rho'}_1$, it then suffices to show \begin{equation}\label{eq:<rhoiL-}
    \frac{1}{2} \le \norm{ U_{\rm L} \rho_{i, \rm L} U_{\rm L}^\dagger - \tr_{\rm R}\lr{ U_{\rm 0}^\dagger \rho_f   U_{\rm 0}} }_1.
\end{equation}

$\rho_{i, \rm L}$ is the restriction of $\rho_i$ in the left half chain that contains $L$ identities, so the density matrix has $2^L$ nonzero eigenvalues (Schmidt values) that are all equal to $2^{-L}$. In other word, there is a basis $\{\ket{j}:j=1,\cdots,2^{2L}\}$ such that \begin{equation}\label{eq:jbasis}
    U_{\rm L} \rho_{i, \rm L} U_{\rm L}^\dagger= 2^{-L}\begin{pmatrix}
        I & 0 \\ 0 & 0
    \end{pmatrix},
\end{equation} 
where each block is $2^L\times 2^L$.
On the other hand, $\rho_f$ has $L+1$ identities in the left half chain, and $U_0$ acts on an all-zero region of $\rho_f$ (across the left and right half chain). As a result \begin{equation}\label{eq:norm<2L}
    \norm{\tr_{\rm R}\lr{ U_{\rm 0}^\dagger \rho_f   U_{\rm 0}}} = 2^{-L-1} \norm{I\otimes \tr_{\rm R}\lr{U_{\rm 0}^\dagger \ket{0}\bra{0} U_{\rm 0}} }\le 2^{-L-1},
\end{equation} 
because the operator inside the norm is a density matrix $\rho$ with $\norm{\rho}\le1$. In words, this holds because in region R $\rho_f$ already is a mixed state with $L+1$ bits, and the unitary $U_0$ can only further increase the entanglement between R and L as it acts only on the sites where $\rho_f$ is pure.

Such spectral difference leads to large $1$-norm of operator $O:=U_{\rm L} \rho_{i, \rm L} U_{\rm L}^\dagger - \tr_{\rm R}\lr{ U_{\rm 0}^\dagger \rho_f   U_{\rm 0}}$, because \begin{align}\label{eq:rho-rho>1/2}
    \norm{O}_1 &= \max_V \abs{\tr OV} \ge \sum_j \abs{O_{jj}} \ge \sum_{j\le 2^L} O_{jj} \nonumber\\
    & = \sum_{j\le 2^L } 2^{-L} - \mlr{\tr_{\rm R}\lr{ U_{\rm 0}^\dagger \rho_f   U_{\rm 0}}}_{jj} \nonumber\\
    &\ge \sum_{j\le 2^L } (2^{-L} - 2^{-L-1}) = \frac{1}{2}.
\end{align}
Here in the first inequality, we have chosen the particular unitary $V$ that is diagonal in the basis \eqref{eq:jbasis}, with the $j$-th diagonal element $V_{jj}=\mathrm{sign}(O_{jj})$. In the second line, we have used \eqref{eq:jbasis} in this basis. The last line comes from $\mlr{\tr_{\rm R}\lr{ U_{\rm 0}^\dagger \rho_f   U_{\rm 0}}}_{jj}\le \norm{ \tr_{\rm R}\lr{ U_{\rm 0}^\dagger \rho_f   U_{\rm 0}} }$ and \eqref{eq:norm<2L}.
This proves \eqref{eq:<rhoiL-}, which then leads to \eqref{eq:UrhoU-rho>}.
\end{proof}

\subsection{Long-range interactions}
Having gained the key intuition for why the shift is hard to implement using local interactions, we now turn to our main goal, namely the study of implementing shift using power-law interactions. In what follows, we assume (\ref{eq:powerlawdecay}) and fix units of time by setting $K=1$.

To proceed, we will need the following Lieb-Robinson bound:
\begin{thm}\cite{strictlylinear_KS,tran2021,ourreview}
Assuming (\ref{eq:powerlawdecay}), 
there exists constants $c_{\rm LR},v,\epsilon>0$ where $\epsilon$ can be arbitrarily small, such that for any local operator $A$,
\begin{equation}\label{eq:LRB_power}
    \norm{A(t)-\PP_r A(t)} \le \norm{A} g_\alpha(t,r)
\end{equation}
where \begin{equation}
    g_\alpha(t,r):=c_{\rm LR} \cdot \left\{\begin{array}{ccc}
        t^{2}(r-vt)^{1-\alpha} & (t\le r/v), & \alpha>3 \\
        \lr{\frac{t}{r^{\alpha-2-\epsilon}}}^{\frac{\alpha-1}{\alpha-2}-\frac{\epsilon}{2}} & (t\le r^{\alpha-2-\epsilon}/v), & 2<\alpha\le 3
    \end{array}\right..
\end{equation}
Here $\PP_r A(t)$ is supported only on sites of distance $\le r$ to the initial local operator $A$. 
\end{thm}  

\begin{proof}
\eqref{eq:LRB_power} is contained in the literature except for the point $\alpha=3$, which we focus on here. For any $\epsilon\in(0,1)$, a system with $\alpha=3$ also satisfies \eqref{eq:powerlawdecay} with exponent $\alpha':=3-\epsilon/2\in (2,3)$. Then one can use \eqref{eq:LRB_power} with $\alpha$ set to $\alpha'$ and $\epsilon$ set to $\epsilon':=\epsilon/2$, so that \begin{align}\label{eq:alpha3}
    \norm{A(t)-\PP_r A(t)} &\le c_{\rm LR} \norm{A} \lr{\frac{t}{r^{\alpha'-2-\epsilon'}}}^{\frac{\alpha'-1}{\alpha'-2}-\frac{\epsilon'}{2}} = c_{\rm LR} \norm{A} \lr{\frac{t}{r^{\alpha-2-\epsilon}}}^{\frac{\alpha'-1}{\alpha'-2}-\frac{\epsilon'}{2}} = c_{\rm LR}' \norm{A} \lr{\frac{vt}{r^{\alpha-2-\epsilon}}}^{\frac{\alpha'-1}{\alpha'-2}-\frac{\epsilon'}{2}} \nonumber\\
    &\le c_{\rm LR}' \norm{A} \lr{\frac{vt}{r^{\alpha-2-\epsilon}}}^{\frac{\alpha-1}{\alpha-2}-\frac{\epsilon}{2}}, \quad \mathrm{for} \quad t\le r^{\alpha'-2-\epsilon'}/v= r^{\alpha-2-\epsilon}/v.
\end{align}
Here $c_{\rm LR}'=c_{\rm LR} v^{-\frac{\alpha'-1}{\alpha'-2}+\frac{\epsilon'}{2}}$ in the first line, and we have used the base $vt/r^{\alpha-2-\epsilon}\le 1$ and $\frac{\alpha'-1}{\alpha'-2}\ge \frac{\alpha-1}{\alpha-2}$ in the second line. \eqref{eq:alpha3} then becomes \eqref{eq:LRB_power} at $\alpha=3$ by redefining the constant $c'_{\rm LR}v^{\frac{\alpha-1}{\alpha-2}-\frac{\epsilon}{2}}\rightarrow c_{\rm LR}$.
\end{proof}

Theorem \ref{thm:shift} can now be generalized to power law interactions:

\begin{thm}\label{thm:T>power}
For power-law interaction \eqref{eq:powerlawdecay}, if $L\ge L_0$ for some constant $L_0$, and the evolution time $T$ satisfies \begin{equation}\label{eq:T>power}
    vT\le \left\{\begin{array}{cc}
       L-C, & \alpha\ge 4 \\
       C L^{\frac{\alpha-1}{3}},  & 3<\alpha<4 \\
       C L^{\frac{[\alpha-1-\epsilon(\alpha/2-1)](\alpha-2-\epsilon)}{2\alpha-3-\epsilon(\alpha/2-1)}},  & 2<\alpha\le 3
    \end{array}\right.
\end{equation}
for  any $\epsilon>0$, and constant $0<C<\infty$ that may depend on $\epsilon$, then \eqref{eq:U-shift>} holds: namely, $U$ is far from $U_{\mathrm{sh}}$.
\end{thm}

\begin{proof}
Following the proof of Theorem \ref{thm:shift}, we construct a circuit approximation $\tU$ still of the form \eqref{eq:tU=}, as shown by the four blue blocks in Fig.~\ref{fig:shift}. \eqref{eq:tU-S>} continues to hold because it only depends on the geometry of the circuit $\tU$. Therefore, due to \eqref{eq:U-S>tU}, we only need to prove $\tU$ is a good approximation for $U$, i.e. \eqref{eq:U-tU}, in time window \eqref{eq:T>power}.

We follow \cite{Tran_2019_polyLC}, which generalizes \cite{HHKL} to ``circuitize'' power-law Hamiltonian dynamics. The complication is that when cutting the chain open between site $0$ and $1$, there are many couplings $H_{xy}$ to deal with instead of just one $H_{01}$ in the local case. We separate the couplings into two groups.

The first group of couplings are highly long-range ones that we can ignore directly.
More precisely, we ignore Hamiltonian terms \begin{equation}\label{eq:Hfar=}
    H_{\rm far} = \sum_{\substack{1\le x \le 2L, 2L<y\le N:\\ \mathsf{d}(x,y)\ge \ell}} H_{xy},
\end{equation}
that couple the left and right half-chain and act on spins of distance $\ge \ell$. The remaining Hamiltonian terms generate $U_{-\rm far}:=\tilde{\mT} \e^{-\ii \int^T_0 \ud t \mlr{H(t)-H_{\rm far}(t) } }$. To choose $\ell$, we bound the error in a similar way as \eqref{eq:UUop-U0<}: \begin{align}\label{eq:Hfar<}
\norm{U- U_{-\rm far}} &\le 
    T\norm{H_{\rm far}}\le T\sum_{\substack{1\le x \le 2L, 2L<y\le N:\\ \mathsf{d}(x,y)\ge \ell}} \norm{H_{xy}} \nonumber\\
    &\le 2T \sum_{x>0, y<0: x-y\ge \ell } (x-y)^{-\alpha} = 2T\sum_{r\ge \ell} r\cdot r^{-\alpha} \le c_\alpha T \ell^{2-\alpha},
\end{align}
for a constant $c_\alpha$ if $\alpha>2$.
To get the second line, we have used \eqref{eq:powerlawdecay} and relaxed the sum to $2$ cuts each for an infinite chain. We therefore choose \begin{equation}
    \ell := (16c_\alpha T)^{1/(\alpha-2)},
\end{equation}
such that \begin{equation}\label{eq:U-Ufar<}
    \norm{U-U_{-\rm far}}\le \frac{1}{16}.
\end{equation}
Moreover, assuming \eqref{eq:T>power}, one can verify that $\ell\lesssim L^{\beta}$ for some $\beta<1$, for all $\alpha$ regimes considered. Therefore, at sufficiently large $L\ge L_0$, \eqref{eq:T>power} implies
\begin{equation}\label{eq:l<L}
    \ell \le 0.1 L.
\end{equation}
In other words, one can always ignore the longest-scale $\sim L$ couplings.

Due to \eqref{eq:U-Ufar<} and triangle inequality, it remains to prove \begin{equation}\label{eq:Ufar-tU<}
    \norm{U_{-\rm far}-\tU}\le \frac{1}{16}.
\end{equation}
Since there are two cuts between left and right, we focus on one cut to show \begin{equation}\label{eq:Ufar-U0<}
    \norm{U_{-\rm far} - U_0 U_{\rm op}} \le \frac{1}{32},
\end{equation}
where $U_{\rm op}$ is generated by open-boundary Hamiltonian $H-H_{\rm far}-H_{\rm cut}$ with
\begin{equation}
    H_{\rm cut} := \sum_{\substack{x<0, y>0:\\ y-x< \ell}} H_{xy}.
\end{equation}
Following the proof of Lemma \ref{lem:HHKL} where $H$ is assumed time-independent for simplicity, \begin{align}\label{eq:Ufar-U0<T}
    \norm{U_{-\rm far} - U_0 U_{\rm op}} &\le T \max_t \norm{H_{\rm cut}(t)-\PP_S H_{\rm cut}(t)} \le T\max_t \sum_{\substack{x<0, y>0:\\ y-x< \ell}} \norm{H_{xy}(t)-\PP_S H_{xy}(t)} \nonumber\\
    &\le T \sum_{\substack{x<0, y>0:\\ y-x< \ell}} (y-x)^{-\alpha} g(T,0.9L)  \le c_\alpha' T g(T,0.9L),
\end{align}
for some constant $c_\alpha'$ because the sum over $x,y$ converges for $\alpha>2$.
Here in the first line, $\PP_S H_{\rm cut}(t)$ is the part of $H_{\rm cut}(t)$ acting inside $S=\{2-L,\cdots, L-1\}$. The second line comes from \eqref{eq:LRB_power}, \eqref{eq:powerlawdecay} and that the distance from any $H_{xy}$ included to the outside of $S$ is larger than $L-\ell\ge 0.9 L$ from \eqref{eq:l<L}. 
We then demand $c'_\alpha T g_\alpha(T,0.9L)\le 1/32$, which for $\alpha\in (2,3]$ becomes \begin{equation}
    T^{1+\frac{1}{\frac{\alpha-1}{\alpha-2}-\frac{\epsilon}{2}}} \le C' L^{\alpha-2-\epsilon},
\end{equation}
for some constant $C'$, which reduces to the third line of \eqref{eq:T>power} for some constant $C$. The first two lines work similarly, because for $3<\alpha$, the error is large so long as $T^3/(L-vT)^{\alpha-1}$ is small.

To summarize, \eqref{eq:T>power} guarantees \eqref{eq:Ufar-U0<}, which implies \eqref{eq:Ufar-tU<} by treating the other cut similarly, and further \eqref{eq:U-tU} by combining \eqref{eq:U-Ufar<}. Finally, this yields \eqref{eq:U-shift>} following the proof of Lemma \ref{lem:HHKL}.
\end{proof}

The bound \eqref{eq:T>power} is probably not qualitatively tight for realizing the shift unitary at $\alpha<4$, since we only compared the unitary action on a \emph{particular} state $\rho_i$.
In the next section, we will present tighter bounds than \eqref{eq:T>power} at sufficiently small $\alpha$. Conceptually, the previous proof argues that entanglement (or, purity) stored in identities cannot be transported easily, and it is conjectured in \cite{Chen_1d21} that entanglement is governed by Frobenius light cone for power law interactions, which cares about \emph{all} states; this conjecture is compatible with \cite{gong2017}.

For the specific task of mapping $\rho_i$ to $\rho_f$, i.e. transferring an identity local density matrix, it can be achieved by state-transfer protocols $U$: \begin{equation}
    U (\alpha \ket{0}_x + \beta \ket{1}_x)\otimes \ket{\bm{0}}_{x+1,\cdots,y} = \ket{\bm{0}}_{x,\cdots,y-1}\otimes (\alpha \ket{0}_y + \beta \ket{1}_y),
\end{equation}
that transfers a qubit information from $x=-L$ to $y=L$ in a background of $0$s.
Such protocols have been constructed \cite{Tran_GHZ21} that asymptotically matches the Lieb-Robinson bound \eqref{eq:LRB_power}; however, these protocols do not saturate \eqref{eq:T>power} because of the extra $T$ factor in \eqref{eq:Ufar-U0<T}. As stated in the introduction, we conjecture that our bounds are not optimal.  

\section{Frobenius perspective}
As we now show, Theorem \ref{thm:T>power} is not optimal at \emph{small} $\alpha$.   To show this result, we will need to take a rather different perspective, based on the Frobenius light cone \cite{Chen_1d21}.  Intuitively, Frobenius bounds study the Frobenius norm, which we define as \begin{equation}
    \lVert A\rVert_{\mathrm{F}} = \sqrt{\frac{\mathrm{tr}(A^\dagger A)}{\mathrm{tr}(1)}}.
\end{equation}
This object may be interpreted more physically as bounding the size of \emph{typical} matrix elements of $A$, rather than the largest matrix element of $A$ possible (the latter is the definition of the operator norm $\lVert A\rVert$). If we can successfully implement $U_{\mathrm{sh}}$, we must implement it in a typical state, so it is natural to expect that the time needed to implement a shift operator can be bounded by studying Frobenius light cones.

\subsection{Operator growth formalism}\label{sec:o_grow}
We now introduce a useful notation for developing these Frobenius bounds.  The key is to view operator growth as a quantum walk in the Hilbert space of operators.
An operator basis for the operator space is Pauli strings, with local basis on each site being the four Pauli matrices denoted by $\{\keto{I}, \keto{X}, \keto{Y}, \keto{Z} \}$. Inner product between operators and the induced Frobenius norm are defined by \begin{equation}\label{eq:innerO}
    \lr{O|O'} := 2^{-N} \tr(O^\dagger O') ,
\end{equation}
where the normalization ensures that each Pauli string has norm $1$.  Notice that $\norm{O}_{\rm F} = \sqrt{\lr{O|O}}$.  For comparison, in this section we denote the usual operator norm as $\norm{O}_\infty \ge \norm{O}_{\rm F}$.

Superoperators are defined as linear maps acting within this operator Hilbert space. We have already encountered the superprojector $\PP$ in \eqref{eq:Pr=}. In addition, we will use the super-identity on a site \begin{equation}
    \mI = \keto{I}\brao{I} + \keto{X}\brao{X}+\keto{Y}\brao{Y}+\keto{Z}\brao{Z},
\end{equation}
which makes $\mI/4$ into a super-density-matrix. The evolution superoperators \begin{equation}
    \mU\keto{O} = \keto{U^\dagger O U}, \quad \mU_{\rm sh}\keto{O} = \keto{U_{\rm sh}^\dagger O U_{\rm sh}},
\end{equation}
are still unitary because $\brao{O}\mU^\dagger \mU\keto{O} = 2^{-N} \tr \mlr{ (U^\dagger O U)^\dagger U^\dagger O U} = \lr{O|O}$. The distance between two superunitaries $\mU,\mathcal{V}$ can be quantified by e.g. $\infty$-norm induced by the operator Frobenius norm: \begin{equation}\label{eq:infF=}
    \norm{\mU-\mathcal{V}}_{\infty, \mathrm{F}} := \max_{O: \norm{O}_{\rm F}=1} \norm{(\mU-\mathcal{V})\keto{O}}_{\rm F}.
\end{equation}
As its dual, the trace distance between two super-density-matrices $\mR,\mR'$ is \begin{equation}
    \norm{\mR-\mR'}_{1, \mathrm{F}}:= \max_{\mathcal{O}: \norm{\mathcal{O}}_{\infty, \mathrm{F}}=1} \mathrm{Tr}[\mathcal{O}(\mR-\mR')],
\end{equation}
where $\mathrm{Tr}$ is the super-trace in operator space.

The evolution \begin{equation}
    \mU=\e^{\int^T_0 \ud t\, \LL(t)}
\end{equation}
is generated by Liouvillian $\LL(t)=\sum_a \LL_a(t)$ where \begin{equation}
    \LL_a(t) \keto{O} := \ii \keto{[H_a(t),O]},
\end{equation}
with $H_a$ being the local terms of $H$.

\subsection{Frobenius bound on realizing shift with long-range interactions}
For power-law interaction \eqref{eq:powerlawdecay}, now we further assume $K=1$ and \begin{equation}\label{eq:trxy=0}
    \tr[H_{xy}(t)]=0,
\end{equation}
which is satisfied easily by subtracting an identity operator for each term.
There is a Frobenius light cone that can be asymptotically slower than Lieb-Robinson (i.e. the light cone for $\infty$-norm bound):

\begin{thm}\cite{Chen_1d21}
Given \eqref{eq:powerlawdecay},
there exists constants $c_{\rm FB}$ determined by $\alpha$, such that for any local operator $A$,
\begin{equation}\label{eq:FrobLC}
    \norm{A(t)-\PP_r A(t)}_{\rm F} \le \norm{A}_F f_\alpha(t,r) 
\end{equation}
where \begin{equation}
    f_\alpha(t,r):=c_{\rm FB}\, t\cdot \left\{\begin{array}{cc}
        \log( r)/ r, & \alpha>2 \\
        \log^2( r)/r, & \alpha=2 \\
        r^{1-\alpha}, & 1<\alpha<2
    \end{array}\right.
\end{equation}
\end{thm}  
Using this result, we can improve Theorem \ref{thm:T>power} to:
\begin{thm}\label{thm:Frob_shift}
For power-law interaction \eqref{eq:powerlawdecay}, if $L\ge L_0$ for some constant $L_0$, and the evolution time $T$ satisfies \begin{equation}\label{eq:T>Frob}
    T\le C_{\rm FB}\cdot \left\{\begin{array}{cc}
       \sqrt{L/\log L}, & \alpha>2 \\
       \sqrt{L/\log^2 L},  & \alpha=2 \\
       L^{(\alpha-1)/2},  & 1<\alpha<2
    \end{array}\right.
\end{equation}
for some constant $C_{\rm FB}$, then \begin{equation}\label{eq:U-shift>F}
    \norm{U-U_{\rm sh}}_{\rm F} \ge 1/8,
\end{equation} so that $U$ does not realize shift operation.
\end{thm}

Before the proof, let us compare Theorems \ref{thm:Frob_shift} and \ref{thm:T>power}. The distance metric is stronger here because $\norm{U-U_{\rm sh}}_{\rm F}\le \norm{U-U_{\rm sh}}$; namely, in the time window \eqref{eq:T>Frob}, $U$ does not perform shift for a \emph{typical} state.
Comparing the time window, the previous upper bound \eqref{eq:T>power} is parametrically larger (stronger) for sufficiently large $\alpha$, and vice versa for small $\alpha$. In particular, there is an exponential separation at $\alpha\in (0,1)$, where $T$ is algebraic with $L$ in \eqref{eq:T>Frob}, while \eqref{eq:T>power} would become $T=\mathrm{O}(\mathrm{polylog}(L))$ if using the Lieb-Robinson bound $g_\alpha(t,r)\sim \e^{\mu t}r^{-\alpha}$ \cite{Hastings_koma} in this $\alpha$ regime. The critical $\alpha_{\rm c}\in (2,3)$ where \eqref{eq:T>Frob} and \eqref{eq:T>power} coincide is \begin{equation}
    \frac{[\alpha-1-\epsilon(\alpha/2-1)](\alpha-2-\epsilon)}{2\alpha-3-\epsilon(\alpha/2-1)} = \frac{1}{2}-\epsilon', \rarrow \alpha\approx \alpha_{\rm c} = 2+1/\sqrt{2},
\end{equation}
where we have set $\epsilon,\epsilon'\rightarrow 0$. The Frobenius bound Theorem \ref{thm:Frob_shift} is then stronger for $\alpha<\alpha_{\rm c}$; for $\alpha>\alpha_{\rm c}$, the two Theorems are in general not comparable because the metric is different.

\begin{proof}
The proof mimics that of Theorem \ref{thm:shift} and \ref{thm:T>power}, where we basically replace all $\norm{V-V'}$ with Frobenius norm $\norm{V-V'}_{\rm F}$. This is valid because one still has triangle inequality for Frobenius norm \begin{equation}
    \norm{A-B}_{\rm F}\le \norm{A}_{\rm F} + \norm{B}_{\rm F}.
\end{equation}
As one can verify, we only need to prove the following three inequalities. \begin{subequations}\label{eq:three}
    \begin{align}
        \norm{U-U_{-\rm far}}_{\rm F} &\le 1/16, \label{eq:U-Ufar<F} \\
        \norm{U_{-\rm far} - \tU}_{\rm F} &\le 1/16, \label{eq:Ufar-tU<F} \\
        \norm{\tU-U_{\rm sh}}_{\rm F} &\ge 1/4. \label{eq:tU-S>F}
    \end{align}
\end{subequations}
Here $U_{-\rm far}$ is generated by $H-H_{\rm far}$ with $H_{\rm far}$ defined in \eqref{eq:Hfar=} where the parameter $\ell$ is determined shortly. 

We proceed one by one. For \eqref{eq:U-Ufar<F}, similar to \eqref{eq:Hfar<} we have
\begin{align}\label{eq:Hfar<F}
\norm{U- U_{-\rm far}}_{\rm F} &\le 
    T\norm{H_{\rm far}}_{\rm F} \nonumber\\
    &\le T \sqrt{\sum_{\substack{1\le x \le 2L, 2L<y\le N:\\ \mathsf{d}(x,y)\ge \ell}} \norm{H_{xy}}_{\rm F}^2 } \nonumber\\
    &\le 2T \sqrt{\sum_{x>0, y<0: x-y\ge \ell } (x-y)^{-2\alpha}} = 2T\sqrt{\sum_{r\ge \ell} r\cdot r^{-2\alpha} }\le \tilde{c}_\alpha T \ell^{1-\alpha},
\end{align}
for some constant $\tilde{c}_\alpha$ if $\alpha>1$. Here the second line comes from orthogonality $\lr{H_{xy}|H_{x'y'}}=0$ if $(x,y)\neq (x',y')$ due to \eqref{eq:trxy=0}. To get the third line, we used $\norm{H_{xy}}_{\rm F}\le \norm{H_{xy}}$. We therefore choose \begin{equation}\label{eq:l=Talpha}
    \ell := (16c_\alpha T)^{1/(\alpha-1)},
\end{equation}
such that \eqref{eq:U-Ufar<F} holds. Moreover, one can verify that \eqref{eq:T>Frob} implies $\ell\le 0.1L$ \eqref{eq:l<L} for sufficiently large $L\ge L_0$.

\eqref{eq:Ufar-tU<F} just comes from \eqref{eq:Ufar-U0<T} with $\infty$-norms changed to Frobenius norm, $g_\alpha(t,r)$ substituted by $f_\alpha(t,r)$, and direct computation of $Tf_\alpha(T,0.9L)=\mathrm{O}(1)$.

For \eqref{eq:tU-S>F}, observe that in the proof of Lemma \ref{lem:shift_rho}, we only demand the $S:=\{1-L,\cdots,L\}$ subsystem in $\rho_i$ to be a pure direct product state (see arguments around \eqref{eq:jbasis} and \eqref{eq:norm<2L}). Therefore Lemma \ref{lem:shift_rho} actually proves that \begin{equation}\label{eq:UrhozU-rho>}
    \norm{\tU \rho_i(z) \tU^\dagger - \rho_f(z)}_1 \ge \frac{1}{2},
\end{equation}
for any string $z\in \{0,1\}^{2L}$, where \begin{equation}
    \rho_i(z) := \bigotimes_{x=1-L}^L \ket{z_x}_x\bra{z_x} \otimes \bigotimes_{x=L+1}^{3L} \frac{I_x}{2},
\end{equation}
(so that the previous $\rho_i=\rho_i(0\cdots0)$), and $\rho_f(z):=U_{\rm sh}\rho_i(z) U_{\rm sh}^\dagger$. We need to relate \eqref{eq:UrhozU-rho>} to \eqref{eq:tU-S>F}.

For each $z$ string, one can purify the system by a $2^{2L}$-dimensional ancilla (labeled by $\rm A$) so that \begin{equation}
    \rho_{i}(z) = \tr_{\rm A}\ket{\Psi_{i}(z)} \bra{\Psi_{i}(z)}.
\end{equation}
\eqref{eq:UrhozU-rho>} then bounds the fidelity between pure states $\tU \ket{\Psi_i(z)}$ and $\ket{\Psi_f(z)}:=U_{\rm sh}\ket{\Psi_i(z)}$ (here $\tU$, for example, is shorthand notation for $\tU\otimes I_{\rm A}$ acting on the enlarged Hilbert space) \begin{align}
    2\lr{1-\abs{\bra{\Psi_i(z)}U_{\rm sh}^\dagger \tU \ket{\Psi_i(z)}}} &\ge 1-\abs{\bra{\Psi_i(z)}U_{\rm sh}^\dagger \tU \ket{\Psi_i(z)}}^2 \nonumber\\
    &= \frac{1}{4} \norm{\tU \ket{\Psi_i(z)} \bra{\Psi_i(z)}\tU^\dagger - U_{\rm sh} \ket{\Psi_i(z)} \bra{\Psi_i(z)}U_{\rm sh}^\dagger}_1^2 \nonumber\\
    &\ge \frac{1}{4}\norm{\tU \rho_i(z) \tU^\dagger - \rho_f(z)}_1^2 \ge \frac{1}{16}. \label{eq:321}
\end{align}
Here the first line comes from $1+\abs{\bra{\Psi_i(z)}U_{\rm sh}^\dagger \tU \ket{\Psi_i(z)}}\le 2$, the second line comes from the relation between fidelity and trace-norm for pure states, and the last line comes from monotonicity under partial trace over ancilla, together with \eqref{eq:UrhozU-rho>}.  

On the other hand, the target \begin{align}\label{eq:F2>fidelity}
    \norm{\tU-U_{\rm sh}}_{\rm F}^2 &= \norm{\tU}_{\rm F}^2 + \norm{U_{\rm sh}}_{\rm F}^2 - \mlr{(U_{\rm sh}| \tU) +{\rm c.c.}} = 2 - 2^{-4L} \mlr{\tr\lr{U_{\rm sh}^\dagger \tU} + {\rm c.c.}} \nonumber\\
    &= 2 - 2^{-4L} \mlr{\sum_z \tr\lr{(I_{S^{\rm c}}\otimes \ket{z}_S\bra{z}) U_{\rm sh}^\dagger \tU} + {\rm c.c.}} \nonumber\\
    &\ge 2\cdot 2^{-2L}\sum_z 1- \abs{2^{-2L}\tr\lr{(I_{S^{\rm c}}\otimes \ket{z}_S\bra{z}) U_{\rm sh}^\dagger \tU}} \nonumber\\
    &= 2\cdot 2^{-2L}\sum_z 1-\abs{\bra{\Psi_i(z)}U_{\rm sh}^\dagger \tU \ket{\Psi_i(z)}} \ge 2\cdot \frac{1}{32} = \frac{1}{16}.
\end{align}
Here all traces work in the physical system without the ancilla. We have used $\norm{V}_{\rm F}=1$ for any unitary $V$ in the first line of \eqref{eq:F2>fidelity}, and inserted $I=I_{S^{\rm c}}\otimes \sum_z\ket{z}_S\bra{z}$ in the second line, where $S^{\rm c}=\{L+1,\cdots, 3L\}$ is the complement of $S$. To get the last line, observe that when computing the fidelity in the enlarged Hilbert space, one can first trace over ancilla because $U_{\rm sh}^\dagger \tU$ does not act on it, which leaves density matrix $2^{-2L}I_{S^{\rm c}}$ in $S^{\rm c}$ that exactly matches the third line. \eqref{eq:F2>fidelity} is just \eqref{eq:tU-S>F} by taking square root; the inequality is simply (\ref{eq:321}).

To summarize, we have proven all inequalities in \eqref{eq:three}, and plugging them into the proof of Theorem \ref{thm:T>power} establishes \eqref{eq:U-shift>F} by considering Frobenius norms throughout. 
\end{proof}

\subsection{The Super2 formalism: an alternate proof}
Here we present an alternative proof that shift cannot be realized in time \eqref{eq:T>power}. The idea is to more seriously promote from the state Hilbert space to the operator space, as described in Section \ref{sec:o_grow}. Note that similar ideas also appear in bounding Frobenius light cones \cite{Chen_1d21}, and quantifying information scrambling \cite{Zanardi_scramble22,Zanardi23}. For example, the super-density-matrices defined later in \eqref{eq:Ri} correspond to subalgebras to be scrambled in \cite{Zanardi_scramble22}.

\begin{thm}\label{thm:super2}
Given the same conditions in Theorem \ref{thm:Frob_shift}, the superunitaries are far: There exists an operator $O$ with \begin{equation}\label{eq:Oinf=1}
    \norm{O}_\infty=1,
\end{equation}
such that \begin{equation}\label{eq:mU-shift>F}    
\norm{\lr{\mU-\mU_{\rm sh}}\keto{O}}_{ {\rm F}} \ge \frac{1}{8}.
\end{equation} 
Therefore, $U$ does not realize the shift operation.
\end{thm}

Theorem \ref{thm:super2} implies \ref{thm:Frob_shift} with a weaker constant $1/8\rightarrow1/16$ in \eqref{eq:U-shift>F}, because \begin{align}\label{eq:mU<U}
    \norm{\lr{\mU-\mU_{\rm sh}}\keto{O}}_{ {\rm F}} &= \norm{U^\dagger O U - U_{\rm sh}^\dagger O U_{\rm sh}}_{\rm F} \le \norm{U^\dagger O \lr{U- U_{\rm sh}}}_{\rm F} + \norm{\lr{U^\dagger- U_{\rm sh}^\dagger} O U_{\rm sh}}_{\rm F} \nonumber\\
    &\le \norm{U^\dagger}_\infty \norm{O}_\infty \norm{U- U_{\rm sh}}_{\rm F} + \norm{U^\dagger- U_{\rm sh}^\dagger}_{\rm F} \norm{O}_\infty \norm{U_{\rm sh}}_\infty = 2 \norm{U- U_{\rm sh}}_{\rm F},
\end{align}
where we have used triangle inequality in the first line, and H\"{o}lder's inequality $\norm{AB}_{\rm F}\le \norm{A}_\infty \norm{B}_{\rm F}$ in the second line.
\eqref{eq:mU-shift>F} also implies the more compact version $\norm{\mU-\mU_{\rm sh}}_{\infty, {\rm F}} \ge 1/8$ from $\norm{O}_{\rm F}\le \norm{O}_\infty$ and the definition \eqref{eq:infF=}.

\begin{proof}
The idea is to promote the proof of Theorem \ref{thm:shift} from state (density matrix) space to operator (super-density-matrix) space, since the shift $U_{\rm sh}$ is also the super-shift $\mU_{\rm sh}$. We call this approach the Super2 picture. The benefit is that some steps in the proof of Theorem \ref{thm:shift} (e.g. Lemma \ref{lem:shift_rho}) generalizes trivially after promotion, while now we can use the slower Frobenius light cone \eqref{eq:FrobLC}  as in Theorem \ref{thm:Frob_shift}. 

Since triangle inequality still holds for the function $\norm{\cdot\keto{O}}_{\mathrm{F}}$ of superoperators, to prove \eqref{eq:mU-shift>F} it suffices to show the analog of \eqref{eq:three}: \begin{subequations}\label{eq:three'}
    \begin{align}
        \norm{\lr{\mU-\mU_{-\rm far}} \keto{O}}_{\rm F} &\le 1/16, \label{eq:U-Ufar<F1} \\
        \norm{\lr{\mU_{-\rm far} - \tmU} \keto{O}}_{\rm F} &\le 1/16, \label{eq:Ufar-tU<F1} \\
        \norm{\lr{\tmU-\mU_{\rm sh}} \keto{O}}_{\rm F} &\ge 1/4, \label{eq:tU-S>F1}
    \end{align}
\end{subequations}
for the to-be-determined operator $O$.
Here $\mU_{-\rm far}, \tmU$ are the superoperator versions of $U_{-\rm far}, \tU$ constructed in the proof of Theorem \ref{thm:Frob_shift}. In particular, $\tmU$ has the circuit geometry in Fig.~\ref{fig:shift}. 

We first prove \eqref{eq:U-Ufar<F1} and \eqref{eq:Ufar-tU<F1} for any $O$ normalized by \eqref{eq:Oinf=1}. Similar to \eqref{eq:mU<U}, we have \begin{equation}\label{eq:mU<U2}
    \norm{\lr{\mU- \mU_{-\rm far}} \keto{O}}_{\rm F} \le 2 \norm{U-U_{-\rm far}}_{\rm F}.
\end{equation}
Following the proof of \eqref{eq:U-Ufar<F}, one has $\norm{U-U_{-\rm far}}_{\rm F}\le 1/32$ with a slightly different constant $1/16\rightarrow 1/32$ in \eqref{eq:l=Talpha}, thus \eqref{eq:mU<U2} implies \eqref{eq:U-Ufar<F1}. \eqref{eq:Ufar-tU<F1} follows similarly from adjusting the constant $1/16\rightarrow 1/32$ in \eqref{eq:Ufar-tU<F}.

It remains to find an operator $O$ with \eqref{eq:Oinf=1} that satisfies \eqref{eq:tU-S>F1}. 
Analogous to \eqref{eq:rhoi} and \eqref{eq:rhof}, we introduce super-density-matrices \begin{equation}\label{eq:Ri}
    \mR_i := \bigotimes_{x=1-L}^L \keto{I}_x\brao{I} \otimes \bigotimes_{x=L+1}^{3L} \frac{\mI_x}{4},
\end{equation}
so that \begin{equation}\label{eq:Rf}
    \mR_f:= \mU_{\rm sh}\mR_i \mU_{\rm sh}^\dagger = \bigotimes_{x=-L}^{L-1}  \keto{I}_x\brao{I} \otimes \bigotimes_{x=L}^{3L-1} \frac{\mI_x}{4}.
\end{equation}
The super-identity region is frozen in dynamics, as we used previously; now the other $\keto{I}$ region is also frozen from \eqref{eq:LI=0}. Nevertheless, we do not need this extra observation: Lemma \ref{lem:shift_rho} generalizes verbatim to \begin{equation}\label{eq:URU-R>}
    \norm{\tmU \mR_i \tmU^\dagger - \mR_f}_{1,\rm F} \ge 1/2,
\end{equation}
since the operator space is just isomorphic to the state space of qudits.

Moreover, the super-density-matrix has a nice property $\mR_i=4^{-2L}\sum_P \keto{P}\brao{P}$ expanded as $4^{2L}$ Pauli strings $P$, so that \eqref{eq:URU-R>} becomes \begin{align}\label{eq:<maxP}
    1/2 &\le 4^{-2L} \norm{\sum_P \tmU \keto{P}\brao{P} \tmU^\dagger - \mU_{\rm sh} \keto{P}\brao{P} \mU_{\rm sh}^\dagger}_{1, \rm F} \le 4^{-2L} \sum_P \norm{\tmU \keto{P}\brao{P} \tmU^\dagger - \mU_{\rm sh} \keto{P}\brao{P} \mU_{\rm sh}^\dagger}_{1, \rm F} \nonumber\\
    &\le 4^{-2L} \sum_P 2 \norm{(\tmU-\mU_{\rm sh})\keto{P}}_{\rm F} \le 2\max_P \norm{(\tmU-\mU_{\rm sh})\keto{P}}_{\rm F} \nonumber\\
    &:= 2\norm{(\tmU-\mU_{\rm sh})\keto{O}}_{\rm F}.
\end{align}
Here the first line follows from triangle inequality. To get the second line of \eqref{eq:<maxP}, we have used \begin{align}
    \norm{\keto{O_1}\brao{O_1}-\keto{O_2}\brao{O_2}}_{1, \mathrm{F}} &\le \norm{\mlr{\keto{O_1}-\keto{O_2}}\brao{O_1}}_{1, \mathrm{F}} + \norm{\keto{O_2}\mlr{\brao{O_1}-\brao{O_2}}}_{1, \mathrm{F}} \nonumber\\
    &= \norm{\keto{O_1}-\keto{O_2}}_{\rm F} \norm{O_1}_{\mathrm{F}} + \norm{O_2}_{\mathrm{F}} \norm{\keto{O_1}-\keto{O_2}}_{\rm F} \le 2\norm{\keto{O_1}-\keto{O_2}}_{\rm F},
\end{align}
where $\keto{O_1}=\tmU\keto{P}$, $\keto{O_2}=\mU_{\rm sh}\keto{P}$, the first line uses triangle inequality, and the second line uses the Super2 analog of $\norm{\ket{\psi}\bra{\psi'}}_1 = \norm{\ket{\psi}} \norm{\ket{\psi'}}$ together with $\norm{O_1}_{\rm F}=\norm{P}_{\rm F}=1$. Therefore, in the last line of \eqref{eq:<maxP}, the desired $O$ is chosen as the maximum Pauli string $P$ achieving the second line, which satisfies \eqref{eq:Oinf=1} and \eqref{eq:tU-S>F1}. 

This finishes the proof because we have established \eqref{eq:three'}. Note that the difference to the proof of Theorem \ref{thm:Frob_shift}, in addition to adjusting constants, is merely in showing \eqref{eq:tU-S>F} versus \eqref{eq:tU-S>F1}.
\end{proof}

\subsection{Analogy to SPT physics}
Lastly, let us comment on an interesting analogy that follows from the previous subsection.   Returning to our general operator growth formalism, observe that 
 $\LL$ cannot be a general super-Hamiltonian.  Not only is it a sum over local super-Hamiltonians that act non-trivially only on certain subsets: $\mathcal{L} = \sum_S \mathcal{L}_S$ where \begin{equation}
     \mathcal{L}_S |A_S \otimes B_{S^c}) = \mathrm{i}|[H_S,A_S] \otimes B_{S^c}),
 \end{equation}but moreover for any $\LL_S$ acting on set $S$,
 \begin{equation}\label{eq:LI=0}
    \LL_S \keto{I_S \otimes A_{S^c}} = 0.
\end{equation}
$\mathcal{L}$ is moreover real for Hermitian elements $O,O'$: \begin{equation}
    \brao{O}\LL \keto{O'}^* = 2^{-N} \mlr{\tr\lr{O \ii [H, O']} }^* = 2^{-N} \tr\lr{\ii[H, O' ] O} = 2^{-N} \tr\lr{O \ii [H, O']} = \brao{O}\LL \keto{O'}.
\end{equation}
In the crudest sense, these requirements can be understood as the imposition of symmetries on the dynamics in the operator Hilbert space. 
Indeed, our proof of Theorem \ref{thm:super2} closely resembles the proof that symmetry-preserving Hamiltonians cannot prepare states belonging to non-trivial symmetry-protected topological in finite time \cite{SPT_lineardepth15}. Therein, one considers the evolution of string order parameters under a time-evolution that potentially generates the SPT phase. Exactly as in our proof of Theorem \ref{thm:super2}, only the endpoints of the string order parameter evolve under symmetry-preserving dynamics, which can be used to show that the evolution cannot be done in finite time. 

{From these observations, it is compelling to draw a connection between the hardness of implementing $U_{\mathrm{sh}}$ and the hardness of generating SPT phases.} However, despite the above similarities, the analogy to SPT physics is not completely satisfying. For example, the operator $\mathcal{R}_i$, which plays the same role in our proof that the string order parameter does in the SPT proof, is not obviously related to any unitary symmetry. Furthermore, in the SPT case, the time-evolution should become easy to implement when the symmetries are not preserved, but this is not straightforwardly true in our proof. The rest of this section is meant to give an alternative proof of the hardness of $U_{\mathrm{sh}}$ which more closely mirrors the proof that SPT states are hard to create, {in order to strengthen the conceptual link between $U_{\mathrm{sh}}$ and SPT phases.}

To start, we consider two copies of the 1D many-body Hilbert space, labelling sites on the first (second) copy as $[x]_A$ ($[x]_B$) with $x\in [1,4L]$. The operator which we will want to implement is $U_{\mathrm{sh}}\otimes U_{\mathrm{sh}}^{-1}$, \textit{i.e.} we shift one copy to the left and the second copy to the right. The symmetries we impose are the following. First, we impose a kind of inversion symmetry $R$ which interchanges the two copies and then implements bond-centered inversion of the periodic 1D system, such that 
\begin{equation} \label{eq:inversion}
    R:\begin{cases} [x]_A\mapsto [4L+1-x]_B\\
    [x]_B\mapsto [4L+1-x]_A.
    \end{cases}
\end{equation}
Clearly $U_{\mathrm{sh}}\otimes U_{\mathrm{sh}}^{-1}$ respects this symmetry. The second symmetry is less conventional: we only allow ``seperable'' interactions which do not couple the two chains. This is not a traditional symmetry in the sense that it does not correspond to dynamics that commute with some unitary operator. Nevertheless, it is a well-defined restriction we can place on the dynamics. This unconventional symmetry is the one place where the analogy to SPT physics is not as direct.

It is straightforward to see that implementing $U_{\mathrm{sh}}\otimes U_{\mathrm{sh}}^{-1}$ subject to these symmetries is equally as difficult as implementing $U_{\mathrm{sh}}$ without symmetry constraints since one protocol can be easily used to construct the other. However, the former has a clearer interpretation in terms of SPT physics. First, we observe that $U_{\mathrm{sh}}\otimes U_{\mathrm{sh}}^{-1}$ is easy to implement in the absence of symmetries using a depth-2 circuit of SWAP operations \cite{Farrelly2020}. These SWAPs exchange sites between the two copies, and therefore violate the separability constraint of the dynamics (they do, however, respect the inversion symmetry). {This is in alignment with the fact that SPT phases are easy to create when symmetry is not enforced.}

Now, we can reformulate the proof of Theorem \ref{thm:super2} in a way that is more closely connected to symmetry and string order parameters. The operator whose evolution we consider, generalizing Eq.~\ref{eq:Ri}, is
\begin{equation}
    \mathcal{S}_i = \bigotimes_{x=1-L}^L \frac{I_{[x]_A}}{2}\frac{I_{[x]_B}}{2}\otimes\bigotimes_{x=L+1}^{3L} \mathrm{SWAP}_{[x]_A,[4L+1-x]_B},
\end{equation}
where $\mathrm{SWAP}_{i,j}$ exchanges the two Hilbert spaces at sites $i$ and $j$. This has a clear interpretation in terms of symmetry, since $\mathcal{S}_i$ corresponds to applying the reflection symmetry $R$ to the region $x\in [L+1,3L]$, similar to a string order parameter \cite{Pollmann2012}. { In fact, $\mathcal{S}_i$ is exactly the order parameter that is used to detect SPT phases with inversion symmetry \cite{Pollmann2012}.}

Evolving this under $U_{\mathrm{sh}}\otimes U_{\mathrm{sh}}^{-1}$ gives
\begin{equation}
    \mathcal{S}_f = \bigotimes_{x=1-L}^L \frac{I_{[x-1]_A}}{2}\frac{I_{[x+1]_B}}{2}\otimes\bigotimes_{x=L+1}^{3L} \mathrm{SWAP}_{[x-1]_A,[4L+1-(x+1)]_B}.
\end{equation}
As required, only the boundaries between the identity region and inversion region in $\mathcal{S}_i$ evolve under inversion-symmetric and separable dynamics. Using this, one can straightforwardly prove the same results as in Theorem \ref{thm:super2} using $\mathcal{S}_i$ and $\mathcal{S}_f$ in place of $\mathcal{R}_i$ and $\mathcal{R}_f$ and considering symmetry-preserving dynamics. {Therefore, if not for the unconventional separability condition we enforced, this alternative proof would exactly resemble that of the hardness of generating SPT phases.}

As a final remark on the connection between the shift operation and SPT phases, we note that the shift is in fact an SPT entangler, meaning that acting with shift on a state belonging to a trivial phase can result in a state belonging to a non-trivial SPT phase \cite{SPT_lineardepth15, Chen2023}. Using known results on the hardness of preparing SPT states, this implies that the shift is hard to implement when certain symmetries are enforced. However, as we have demonstrated, the shift is hard to implement even without enforcing symmetries, and this does not follow immediately from its status as an SPT entangler. In a sense, the discussion in this section is meant to interpret this unconditional hardness of implementing the shift in terms of SPT physics. Going forward, we believe that there are more connections to be made between quantum cellular automata and SPT phases.




\section{Outlook}
In this paper, we have proved that the shift unitary $U_{\mathrm{sh}}$ cannot be implemented in a system-size-independent time in one dimension, with power-law interactions that decay faster than $1/r$. This result is not surprising, based on the earlier work \cite{QCA_LRB22} that proved this result for \emph{infinite} one-dimensional lattices.  Our work fills in a missing piece of the story by showing that their result holds, as expected, for finite lattices.   Due to finiteness, our proof relied on somewhat different techniques.

Our work also provides another practical application of the Frobenius light cone \cite{hierarchy20,Chen_1d21} in constraining the time necessary to implement a particular quantum unitary on all states (as opposed to just one state).  It would be interesting to understand if a similar approach applies to e.g. more general tasks of quantum routing \cite{routing_graph23,Friedman:2022vqb}, and if the Frobenius light cone is \emph{still too weak} of a bound on the implementation time of certain unitaries.  Such a result may require a qualitatively new approach to proving the hardness of implementing particular unitaries, which we leave to future work.  It would also be interesting to construct explicit protocols for implementing shift using power-law Hamiltonians, as has been done for state-transfer.

\section*{Acknowledgements}
This work was supported by the Alfred P. Sloan Foundation under Grant FG-2020-13795 (AL), the Department of Energy under Quantum Pathfinder Grant DE-SC0024324 (CY, AL), and the Simons Collaboration
on Ultra-Quantum Matter, Award No. 651440 (DTS).

\bibliography{biblio}

\end{document}